\documentclass[manuscript,screen,review=false,anonymous=false]{acmart}
\acmJournal{TEAC}
\setcopyright{none}
\acmDOI{}
\acmYear{2026}
\usepackage{mathtools}
\newcommand{\R}{\mathbb R}
\newcommand{\E}{\mathbb E}
\newcommand{\OPT}{\operatorname{OPT}}
\newcommand{\URCM}{\mathrm{URCM}}

\newcommand{\med}{\operatorname{med}}
\newcommand{\clip}{\operatorname{clip}}
\newcommand{\ind}{\mathbf 1}
\newcommand{\norm}[1]{\left\lVert#1\right\rVert}

\newcommand{\ip}[2]{\left\langle#1,#2\right\rangle}
\newcommand{\dd}{\,\mathrm d}

\begin{document}
\newtheorem{remark}[theorem]{Remark}
\title[The Exact Approximation Ratio of Rotated Median]{The Exact Approximation Ratio of Uniformly Rotated Coordinate-wise Median in the Euclidean Plane}
\author{Song Zichen}
\authornote{Sole author and corresponding author. Master's degree in Computer Science.}
\affiliation{%
  \institution{City University of Hong Kong}
  \country{China}
}
\email{72610558@cityu-dg.edu.cn}
\renewcommand{\shortauthors}{Song Zichen}

\begin{abstract}
Uniformly rotated coordinate-wise median chooses a random orthonormal coordinate system, takes a median in each coordinate, and maps the resulting point back to the Euclidean plane. We determine its exact worst-case expected approximation ratio when social cost is the $L_p$ norm of the agents' Euclidean distances and $1<p<2$. The ratio is
\(\frac{2^{2-1/p}}{\pi}\int_0^{\pi/2}(\cos^p\theta+\sin^p\theta)^{1/p}\dd\theta\).
This expression was previously established as a lower bound by Chan, Lin, and Wang; our contribution is the matching upper bound. The proof establishes a strengthened coordinate-wise median inequality relative to an arbitrary reference point. Its right-hand side is linear in a sum of direction-dependent norms, which permits direct averaging over rotations without the loss incurred by passing through a $p$th-moment bound. We give all auxiliary inequalities and a self-contained proof of tightness using the established two-cluster-and-outlier construction. The upper bound holds for every finite profile and every measurable choice within the coordinate median intervals, while odd-size profiles suffice for the matching lower bound. The result characterizes this fixed mechanism, rather than the optimal approximation ratio among all randomized strategyproof mechanisms.
\end{abstract}

\maketitle

\section{Introduction}

Consider a public facility whose location must be chosen on the basis of its users' reported locations. A user may prefer the facility to be near a home or workplace, but the planner cannot generally verify that preferred location. A rule that minimizes the reported total cost can therefore create an incentive to exaggerate a location. The mechanism-design question is how much efficiency can be retained when the rule must instead make truthful reporting a best response. In this paper, there are no monetary payments and no capacity constraints. The only decision is a point in the plane, and each user's loss is its Euclidean distance from that point.

For a reader accustomed to approximation algorithms, the relevant restriction is on the decision rule rather than its running time. We compare a simple truthful rule with the unconstrained minimum of a social objective. The comparison is made on the true location profile. The rule need not compute that minimum, and the proof need not explain how agents could collectively implement it. The optimum serves as a benchmark for measuring the efficiency lost by using the prescribed mechanism. Likewise, a bad instance in our analysis is a profile of true locations, not an equilibrium generated by coordinated false reports.

Coordinate-wise median is particularly useful in this setting because its incentive property can be understood one coordinate at a time. Its welfare analysis cannot generally be separated in the same way. Euclidean distance combines the two coordinate displacements through a square root, while a coordinate median is characterized by counts of points on either side of a line. Our argument concerns the relationship between those counting constraints and the resulting geometric cost. Random rotation introduces a second issue: the bound must retain enough directional information to be averaged sharply.

Randomizing the axes addresses this dependence while preserving strategyproofness for every realization of the random choice. The resulting rule, uniformly rotated coordinate-wise median (URCM), is the object of this paper. We study social costs
\begin{equation}\label{eq:cost-intro}
 C_p(X;y)=\left(\sum_{i=1}^n\norm{x_i-y}_2^p\right)^{1/p},
 \qquad 1<p<2.
\end{equation}
The parameter $p$ changes the aggregation of individual distances; the underlying distance remains Euclidean. After normalization by $n^{1/p}$, these objectives interpolate between average distance and root-mean-square distance. Increasing the exponent places greater emphasis on large individual distances. The mechanism itself does not depend on $p$.

There are consequently two distinct uses of a norm in this problem. The norm $\norm{x_i-y}_2$ measures a displacement in the physical plane. The exponent $p$ then aggregates the $n$ individual distances into one social cost. For example, if two candidate facility locations produce distance vectors $(1,1)$ and $(0,2)$, their total distances agree, but their $L_p$ social costs are $2^{1/p}$ and $2$, respectively. The second location becomes relatively less attractive for every $p>1$. This elementary comparison explains why a mechanism that is well understood for total distance still requires a new analysis for intermediate exponents.

The exponent should not be interpreted as a change in agents' geometric preferences. Throughout the paper, an agent prefers whichever facility is closer in Euclidean distance. Only the planner's way of evaluating the entire distance vector changes. This distinction allows the incentive argument to remain fixed while the approximation ratio varies with $p$. It also explains why results for non-Euclidean distance functions do not automatically answer the question studied here.

The performance criterion is the expectation of the realized social cost. Because the social objective is nonlinear, this criterion differs from the $p$th root of the expected $p$th power of the cost. That distinction is the source of the analytical gap addressed here.

\subsection{The remaining gap}

\citet{chan2026} establish the following bounds for URCM, with the worst case taken over all finite profiles and all numbers of agents:
\begin{equation}\label{eq:prior-bounds}
 L(p)\le\alpha_p(\URCM)\le U(p),
\end{equation}
where
\begin{align}
 L(p)&=\frac{2^{2-1/p}}{\pi}\int_0^{\pi/2}
       (\cos^p\theta+\sin^p\theta)^{1/p}\dd\theta,
       \label{eq:L}\\
 U(p)&=2\left(\frac{\Gamma((p+1)/2)}
                     {\sqrt\pi\,\Gamma(1+p/2)}\right)^{1/p}.
       \label{eq:U}
\end{align}
Their Theorem~2 expresses the lower bound as an integral over $[0,\pi/4]$; symmetry of the integrand makes it identical to~\eqref{eq:L}. Their conclusion leaves the tight ratio in the interval $1<p<2$ unresolved. The bounds agree at the endpoints but differ strictly in the interior.

We close this gap by proving that the previously known lower bound is an upper bound as well. Thus the integral, and the extremal construction that gives the lower bound, are not new contributions of this paper. The contribution is a deterministic inequality that is strong enough to certify their optimality for this mechanism.

To interpret the two sides of~\eqref{eq:prior-bounds}, it is useful to separate their quantifiers. The lower bound says that some profiles force a ratio arbitrarily close to $L(p)$. It does not say that most profiles behave this way, and it does not bound the cost on a general input from above. The upper bound says that no profile has a ratio exceeding $U(p)$. Identifying the exact ratio requires these two kinds of statements to meet. Merely finding additional instances close to the existing lower bound cannot establish that meeting point.

The Gamma function in~\eqref{eq:U} is a compact expression for an angular moment of a coordinate projection. Its appearance does not signal a probabilistic assumption on the input locations. The profile is arbitrary and fixed; only the coordinate system is random. The integral in~\eqref{eq:L} retains the two coordinate projections together before averaging. We will show that this difference in the order of aggregation is precisely what separates the two formulas. No change to the mechanism, its distribution over angles, or its incentive guarantees is needed.

\subsection{Result and proof idea}

Our main result is
\begin{equation}\label{eq:main-intro}
 \alpha_p(\URCM)=L(p)\qquad\text{for every }1<p<2.
\end{equation}
The upper bound has no restriction on the number of agents, repeated locations, or the geometry of the input. The lower bound is approached by odd-size profiles with two large equal clusters and one distant outlier. Consequently, the equality is a statement about the supremum over finite instances; it does not assert attainment for a prescribed number of agents.

The main technical result concerns a fixed coordinate system. Let $m$ be a coordinate-wise median and let $o$ be any reference point. Set $r_i=\norm{x_i-o}_2$. We prove
\begin{equation}\label{eq:H-intro}
 C_p(X;m)\le
 2^{1-1/p}
 \frac{\sum_i r_i^{p-1}\norm{x_i-o}_p}
      {(\sum_i r_i^p)^{(p-1)/p}}.
\end{equation}
Here $\norm{\cdot}_p$ is the two-dimensional coordinate norm, rather than the social-cost norm on $\R^n$. The reference point need not minimize any objective. Applying~\eqref{eq:H-intro} in each rotated coordinate system and choosing $o$ to minimize~\eqref{eq:cost-intro} gives~\eqref{eq:main-intro} by linearity of expectation.

The proof of~\eqref{eq:H-intro} uses a homogeneous auxiliary function that combines Euclidean length with the coordinate $p$-norm. We bound the distance to $m$ pointwise by this function and two indicator corrections, one for each coordinate half-plane determined by $m$. The corrections disappear after summation because at least half of the agents lie on each required side of a coordinate median. The remaining expression is optimized over one scalar parameter. A comparison between two convexity remainders supplies the boundary estimate needed for the pointwise inequality.

This argument retains the directional norm inside the expectation. A bound on the expected $p$th power of the cost instead produces $U(p)$ after taking a root. The difference can be small numerically: at $p=3/2$, $L(p)\approx1.352524$, whereas $U(p)\approx1.352999$. The purpose of the result is to identify the exact guarantee and the inequality that certifies it.

\subsection{Related work and scope}

The literature on strategic facility location varies along several independent dimensions: the geometry of feasible locations, the number of facilities, the agents' preferences, the social objective, and the information available to the mechanism. The survey of \citet{survey2021} provides a broad account of these distinctions. Our question fixes all of them except the social-cost exponent: there is one unconstrained facility in the Euclidean plane, agents prefer shorter Euclidean distances, and the mechanism receives only their reported locations. This specificity matters when comparing approximation factors. A constant established for a different geometry or a different order of expectation and aggregation need not bound the objective studied here.

\paragraph{Why median mechanisms are a natural starting point.}
The role of medians in strategyproof choice predates approximation analysis. \citet{moulin1980} study strategyproof choice under single-peaked preferences, and \citet{border1983} investigate straightforward elections, unanimity, and phantom voters on a multidimensional preference domain. These works explain why order statistics and generalized medians arise when monetary transfers are unavailable. Their characterization assumptions are part of their conclusions: a result for a specified preference domain should not be read as a characterization of every mechanism under arbitrary spatial preferences.

For spatial location problems, \citet{peters1992} connect Pareto optimality, anonymity, and strategyproofness, including the role of coordinate medians in the Euclidean plane with an odd number of agents. \citet{peters1993} and \citet{vanderstel2000} examine generalized medians and the geometry of strictly convex norms. This line of work supplies the structural background for our mechanism, but its central question differs from ours. A characterization identifies which rules can satisfy incentive and efficiency requirements; an approximation analysis quantifies the social cost of a particular admissible rule. Our reference-point inequality addresses the second question. We do not strengthen these classical characterizations or assert Pareto optimality for every tie convention considered in our analytic bound.

\paragraph{From exact implementation to approximate welfare.}
\citet{procaccia2013} formulate approximate mechanism design without money as a way to measure the welfare loss imposed by incentive constraints. Within that framework, the relevant benchmark is the unconstrained optimum, even if directly choosing that optimum would invite manipulation. Our approximation ratio uses precisely this comparison. The incentive argument for each fixed orientation is elementary; the substantive issue is how much social cost remains after averaging those orientations.

Geometry can change that issue substantially. \citet{schummer2002} study strategyproof location on networks, while \citet{alon2010} analyze strategyproof approximation of the minimax objective on networks. \citet{meir2019} consider three agents on a circle. These works show why a statement about facility location needs both a feasible space and an objective. Their network and circle settings do not supply the planar angular averaging step used below: here a rotation acts on the entire Euclidean profile and preserves every physical distance.

\paragraph{Cost aggregation versus the underlying spatial norm.}
The most relevant objective-level predecessor is \citet{feigenbaum2017}, who study approximation of the $L_p$ norm of agents' costs for facility location on a line. Their formulation makes the aggregation exponent a parameter of the social objective. Our setting retains this interpretation of $p$ but introduces two-dimensional Euclidean geometry. By contrast, \citet{lin2020} study deterministic two-agent mechanisms in $L_p$ space, where the norm describes spatial distance. These are different uses of a norm parameter. In this paper, $p$ never replaces the Euclidean norm in an agent's distance to the facility.

\citet{feldman2013} analyze strategyproof facility location with the least-squares objective on line and tree networks. Their work is relevant to the interaction between nonlinear welfare and randomization, but a bound on expected squared cost cannot simply be converted into an exact bound on expected square-root cost. The expectation is taken outside the realized social-cost norm in our definition. Similarly, \citet{fotakis2016} study concave individual cost functions, a different way to depart from linear distance costs. We keep individual preferences induced by Euclidean distance and vary the aggregate used by the analyst. These distinctions explain why apparently related nonlinear-cost results do not settle the present approximation ratio.

For the coordinate-wise median itself, \citet{goel2023} analyze approximation and optimality within a class of deterministic, anonymous, strategyproof mechanisms in two dimensions. \citet{gravin2025} study approximation guarantees of the median mechanism in higher-dimensional normed spaces. The latter addresses an important geometric extension, but a dimension-dependent or dimension-free guarantee for total distance is not an exact evaluation of the expected $L_p$ social cost of a randomly oriented planar mechanism.

The closest results are the recent analyses of \citet{chan2026,barak2026,hastings2026}. \citet{chan2026} settle deterministic Euclidean-plane coordinate-median ratios across the $L_p$ spectrum and analyze randomized mechanisms, including URCM. Independently, \citet{hastings2026} obtain tight deterministic coordinate-median bounds in two-dimensional $\ell_q$ spaces for general aggregation exponents. With Euclidean distances and $1<p<2$, the deterministic ratio is $\sqrt2$. \citet{barak2026} analyze random rotation for the sum of Euclidean distances, obtaining $4/\pi$ in the plane, and also consider higher-dimensional and prediction-augmented settings. Their two-cluster-and-outlier construction is a precursor of the lower-bound family used here. The lower-bound target for our exponent range is already present in \citet{chan2026}; our claimed contribution is the matching upper bound. We include the lower-bound argument for completeness, without claiming the construction as new.

\paragraph{What randomization and truthfulness mean here.}
The random orientation is drawn independently of the reports, and each realized rule is strategyproof under the fixed order-statistic convention specified below. Thus the mechanism is universally strategyproof. This property concerns unilateral deviations. It should not be conflated with the coalitional requirements considered in the spatial analyses of \citet{bordes2011} or the group-strategyproof characterizations of \citet{tang2020}. Our upper bound does not require a new incentive characterization and does not establish group strategyproofness.

An expected approximation guarantee also leaves open how variable an individual agent's realized outcome can be. \citet{procaccia2018variance} study approximation--variance tradeoffs in facility location games. Their perspective clarifies the interpretation of our result: an exact worst-case expectation describes one performance criterion, rather than a complete risk assessment of the randomized outcome. In particular, our proof does not claim a sharp variance or tail-probability guarantee.

\paragraph{Boundaries and possible extensions of the model.}
The single-facility assumption removes assignment decisions from the mechanism. With two or more facilities, the interaction between placement and the agents' nearest facilities changes both the incentive analysis and the available approximation bounds; see \citet{lu2010} and \citet{fotakis2014}. Capacity constraints introduce another assignment issue, studied by \citet{aziz2020} and \citet{walsh2022}. The half-plane counting conditions exploited here concern one coordinate median of all reports. They do not by themselves yield a truthful capacitated or multiple-facility mechanism.

There are also objectives that evaluate equity more directly than a norm of distances. \citet{cai2016} study minimax envy, and \citet{walsh2025} investigate equitable facility location using criteria including the Gini index and Nash welfare. These are relevant alternatives when the intended application is fairness between agents. Increasing the aggregation exponent changes how large distances influence $C_p$, but an approximation guarantee for $C_p$ is not automatically a guarantee for these other equity criteria. We therefore interpret our theorem as a precise welfare bound for the specified objective.

Finally, predictions can provide information beyond the submitted reports. \citet{agrawal2024} study learning-augmented facility location mechanisms, and \citet{christodoulou2024} investigate mechanism design with output advice. \citet{barak2024mac} allow location predictions that are mostly and approximately correct, while \citet{balkanski2024} study randomized strategic facility location with predictions. These papers motivate asking whether additional information can improve a mechanism's performance. Our guarantee uses no such information and has no prediction-accuracy parameter. Extending the reference-point argument to an advice-dependent choice of orientation would require accounting for that dependence explicitly; the uniform angular average in our proof cannot simply be assumed to persist.

Taken together, these comparisons locate a narrow contribution within a broader literature. We analyze a fixed, already proposed randomized mechanism under a fixed family of objectives. We neither claim optimality among all randomized strategyproof mechanisms nor replace the separate analyses needed for other geometries, assignments, or information models. Section~\ref{sec:model} specifies the benchmark and quantifiers. Section~\ref{sec:median} supplies the deterministic inequality, Section~\ref{sec:rotation} averages it, and Section~\ref{sec:lower} establishes tightness. Section~\ref{sec:comparison} identifies the strict gap between the exact ratio and the earlier moment-based bound.

\subsection{An accessible route through the argument}

The technical part of the paper can be read as a response to three questions. First, which information about a coordinate median can be used without explicitly solving for its location? The answer is the pair of half-plane counting constraints: in each coordinate, at least half the reports are at or beyond the median and at least half are at or below it. Second, how can those constraints produce an estimate that survives averaging over directions? The answer is a pointwise inequality with indicator corrections. Third, why is the resulting average optimal? The answer is an instance where a low-mass outlier determines the median's coordinate choices while contributing negligibly to normalized social cost.

The proof of the upper bound runs in the opposite direction from a direct search for worst-case profiles. We do not characterize every profile that might maximize the ratio. Instead, we build a certificate that applies to every profile at once. Each input point contributes a term to the certificate, and the median property ensures that the extra terms cancel after summation. This is why the analysis can allow an arbitrary number of agents and arbitrary repeated locations without a case distinction for each possible order of the coordinates.

Readers primarily interested in the performance guarantee can first read Theorem~\ref{thm:reference}, take it as a temporary premise, and proceed to Section~\ref{sec:rotation}. That section explains in detail why the form of the deterministic inequality yields the exact integral. Section~\ref{sec:lower} then identifies profiles showing that the average cannot be reduced. Returning to Section~\ref{sec:median} reveals how the deterministic certificate is constructed. Readers interested in the inequality itself can follow the lemmas in their presented order; each supplies a different component of the pointwise estimate.

The proof uses elementary differentiation, norm comparison, convexity, H\"older's inequality, and dominated convergence. The difficulty is in arranging these tools so that their losses are compatible. In particular, the boundary lemma compares an error introduced by the auxiliary function with a compensating convexity gap. Bounding either quantity separately by a coarse constant would discard the dependence needed to match them. The explanatory paragraphs around the lemmas identify what each estimate must retain and where that information is used later.

One limitation should remain visible throughout this route. Sharpness of a particular mechanism is different from optimality within a class of mechanisms. Our lower-bound family is designed to exploit coordinate medians. It does not force the same loss on a mechanism that chooses a facility by a different rule. The conclusion therefore answers the exact-analysis question for URCM and leaves the broader mechanism-design optimization question open.

\section{Model and Main Theorem}\label{sec:model}

\subsection{Locations, objectives, and the benchmark}

The model has one type of strategic input and one source of randomization. Each agent supplies a point, and the mechanism independently samples an orientation. There is no assumed distribution of agents, no prior over their reports, and no randomness in the profile when the approximation ratio is evaluated. We use an ordered list rather than a set so that several agents can have the same preferred location. This convention is important for the lower bound, where large groups of agents are deliberately placed at identical points.

The optimization variable $y$ ranges over the entire plane. The benchmark facility is therefore not required to coincide with a reported point. Nor is it required to be the centroid, a median, or the realization of another truthful rule. If all locations are translated or all are multiplied by a common positive scale, both the mechanism's realized cost and the optimum transform by the same scale. Their ratio is unchanged. These invariances will permit us to move a reference point to the origin during the deterministic proof without changing the underlying question.

A profile is an ordered list $X=(x_1,\ldots,x_n)\in(\R^2)^n$ with $n\ge1$; repetitions are allowed. Agent $i$ has cost $\norm{x_i-y}_2$ when the facility is placed at $y$. We use~\eqref{eq:cost-intro} as the social objective and write $\OPT_p(X)=\min_y C_p(X;y)$. The minimum exists because the objective is continuous and tends to infinity as $\norm y_2\to\infty$. A zero optimum means that all locations coincide, in which case every coordinate median returns the common location.

The normalization by the optimum has a specific role. A large absolute cost can simply reflect that a profile has large diameter. Dividing by the smallest possible cost removes that scale. For a fixed $p$, the desired upper bound must hold even as the number of agents and the diameter vary independently. In particular, we cannot assume that the farthest agent stays within a bounded distance of the others. The lower-bound construction will exploit exactly that freedom.

\subsection{What a coordinate median specifies}

A coordinate median is determined by order rather than by distances between adjacent reports. If a point already lies to the right of the median, moving it farther right can leave the median unchanged even though it changes a distance-based objective substantially. Conversely, an additional point can change which order statistic is selected while representing a vanishing fraction of the agents. Keeping the counting definition explicit makes both the upper-bound cancellation and the lower-bound construction easier to interpret.

For a finite real sample $z_1,\ldots,z_n$, a median is any $a$ satisfying
\begin{equation}\label{eq:median-count}
 \#\{i:z_i\le a\}\ge n/2,
 \qquad \#\{i:z_i\ge a\}\ge n/2.
\end{equation}
For odd $n$ the median is unique. For even $n$ the medians form the closed interval between the two middle order statistics. Unless stated otherwise, the mechanism selects the lower median, the $\lceil n/2\rceil$th order statistic.

For example, the sample $(0,0,1,1)$ has the whole interval $[0,1]$ as its set of medians under~\eqref{eq:median-count}. The lower-median convention selects $0$. The sample $(0,0,1,1,7)$ instead has the unique median $1$. The value $7$ could be replaced by any larger number without changing that median. Our deterministic upper bound uses only the inequalities in~\eqref{eq:median-count}, so it remains valid for every point in a median interval. The mechanism's incentive property requires more care because an input-dependent choice within that interval need not behave like a fixed order statistic.

In two dimensions, the two scalar medians are chosen independently in the same coordinate system. Their pair need not be one of the reported points. It is also not generally the point minimizing the sum of Euclidean distances. The term ``coordinate-wise'' is therefore essential: it specifies a particular location rule, rather than a generic multivariate notion of a median. Later, $\norm{\cdot}_p$ will refer to the coordinate norm in that same fixed system, and it will rotate when the axes rotate.

\subsection{Random orientation and a small example}

An orthonormal coordinate system consists of two perpendicular unit vectors. The coordinate of a location along an axis is its inner product with that vector. Once the scalar medians have been found, multiplying each by its axis vector and adding reconstructs a physical point. The input profile is not randomly perturbed: only the representation used by the decision rule changes. Since a rotation is an isometry, it leaves all Euclidean distances between fixed physical points unchanged.

For $\theta\in[0,2\pi)$, let
\[
 e_\theta=(\cos\theta,\sin\theta),\qquad
 f_\theta=(-\sin\theta,\cos\theta).
\]
The mechanism draws $\Theta$ uniformly, independently of the reports, computes
\[
 a_\theta=\med_i\ip{x_i}{e_\theta},\qquad
 b_\theta=\med_i\ip{x_i}{f_\theta},
\]
and returns $m_\theta=a_\theta e_\theta+b_\theta f_\theta$.

As a concrete illustration, put two agents at $(1,0)$, two at $(0,1)$, and one at $(0,0)$. In the original axes, each coordinate has three zeros and two ones, so the output is $(0,0)$. Its social cost is $4^{1/p}$. In the axes at angle $\pi/4$, the first projections are four copies of $1/\sqrt2$ and one zero, while the second projections are two copies of $-1/\sqrt2$, one zero, and two copies of $1/\sqrt2$. Their medians are $1/\sqrt2$ and $0$, giving physical output $(1/2,1/2)$. All five agents are then at distance $1/\sqrt2$, so the cost is $5^{1/p}/\sqrt2$.

This calculation illustrates axis dependence, not a worst-case analysis. We have not claimed that either output is the optimum, or that the improvement obtained at this particular angle represents the uniform angular average. It does show why randomizing the axes can matter even though distances themselves are rotation invariant. The location selected by the rule changes under the rotated representation; the metric does not.

\subsection{The order of expectation and aggregation}

For each sampled angle, the rule produces an ordinary deterministic location, and that location produces an ordinary $L_p$ social cost. The expectation is then taken over those realized costs. The same angle is used for every agent. One should therefore not replace the mechanism by a procedure that chooses an independent angle for each individual distance. The coupling of the agents through one common output is part of the problem.

The approximation ratio is
\begin{equation}\label{eq:ratio}
 \alpha_p(\URCM)=
 \sup_{n\ge1}\ \sup_{X:\OPT_p(X)>0}
 \frac{\E_\Theta C_p(X;m_\Theta)}{\OPT_p(X)}.
\end{equation}
The numerator is neither $(\E C_p(X;m_\Theta)^p)^{1/p}$ nor the $L_p$ norm of the vector of expected individual distances.

A simple random distance vector separates the three operations. Suppose, only for illustration, that with equal probabilities the vector is $(0,0)$ or $(2,0)$. Its expected $L_p$ norm is $1$, whereas the $p$th root of its expected $p$th-power cost is $2^{1-1/p}>1$. Now consider a vector equally likely to be $(2,0)$ or $(0,2)$. Its realized $L_p$ norm is always $2$, but the $L_p$ norm of the expected vector $(1,1)$ is $2^{1/p}<2$. These examples are statements about aggregation, not asserted facility-location profiles. They explain why an analysis of one expression is not automatically sharp for the others.

The outer supremum in~\eqref{eq:ratio} also deserves attention. An upper bound on it must work for every finite profile. A matching lower bound can be supplied by a sequence whose number of agents tends to infinity; it need not be attained by a single finite profile. All limits below fix $p$ first. The result does not require one convergence estimate that is uniform over all $p$ approaching the endpoints of the interval.

\subsection{Why truthful reporting survives rotation}

We include the standard incentive argument to separate the mechanism's validity from its approximation analysis. Fixing the angle leaves a deterministic coordinate-wise rule. The crucial fact is that, when all other reports are held fixed, a fixed order statistic clips the remaining report to an interval determined by those other reports. Truthful reporting gives the point of that interval closest to the true coordinate. This fact involves neither the social-cost exponent nor the eventual upper bound.

\begin{proposition}[Standard incentive property]\label{prop:sp}
URCM with the lower-median convention is universally strategyproof: truthful reporting minimizes every agent's distance for each fixed realization of the sampled angle.
\end{proposition}
\begin{proof}
Fix an angle and the other agents' reports. In one coordinate, write their sorted reports as $z_{(1)}\le\cdots\le z_{(n-1)}$ and set $k=\lceil n/2\rceil$. Use $z_{(0)}=-\infty$ and $z_{(n)}=+\infty$. The $k$th order statistic after inserting a report $z$ is its projection onto $[z_{(k-1)},z_{(k)}]$. Reporting the true coordinate therefore minimizes the absolute difference between that coordinate and the output over all reports. This applies separately to both coordinates. Truthful reporting minimizes both squared coordinate displacements, hence their sum and its square root. Rotating back preserves Euclidean distances. For $n=1$, the rule simply returns the report and the same conclusion is immediate.
\end{proof}

The proposition is stronger than truthfulness merely in expectation: for every sampled orientation, no agent can gain by changing its report. The randomization can thus be regarded as choosing one rule from a family of truthful rules before applying it. Independence of the orientation from the reports is important for this interpretation. If agents could alter the distribution of the axes by changing their reports, the fixed-angle argument would not by itself establish the same incentive guarantee.

The social objective is absent from the proof of Proposition~\ref{prop:sp} because an agent cares only about its own distance. Changing $p$ changes how the planner compares profiles of distances, not what counts as a profitable deviation for an individual. The analytical task that remains is therefore entirely a worst-case comparison between truthful outcomes and the unconstrained social optimum.

\begin{theorem}[Exact approximation ratio]\label{thm:main}
For every $1<p<2$, $\alpha_p(\URCM)=L(p)$, with $L(p)$ given by~\eqref{eq:L}. The upper bound holds for any measurable rule selecting a value in each coordinate median interval. For each fixed $p$, a sequence of finite odd-size profiles has approximation ratios converging to $L(p)$.
\end{theorem}

The theorem combines two claims that will be proved separately. The universal upper bound is obtained from a deterministic statement about every coordinate median. The lower bound uses profiles with unique medians, so it is unaffected by how even samples are handled. This asymmetry is useful: tie-breaking cannot invalidate the upper estimate, and it cannot remove the instances that force the constant in the supremum.

The theorem also specifies where the substantive progress lies relative to the earlier bounds. The target value is already identified by a known construction. What is missing from a lower-bound construction is a reason that a different, less symmetric profile cannot be worse. The next section supplies that reason in the form of a reference-point inequality. Its statement may initially look more complicated than a constant-factor approximation, but the additional structure is what allows the random orientation to be used exactly.

The broader tie-breaking statement in Theorem~\ref{thm:main} concerns approximation only. It does not extend Proposition~\ref{prop:sp} to arbitrary report-dependent selections within median intervals. Measurability ensures that the expected cost is defined; for a fixed finite profile, all such outputs lie in a bounded set.

\section{A Strengthened Coordinate-wise Median Inequality}\label{sec:median}

We first establish the deterministic estimate that makes the rotation average sharp.

\subsection{What must be preserved before averaging}

For an ordinary approximation bound, it might seem natural to seek a constant $K$ such that $C_p(X;m)\le K C_p(X;o)$ for every reference point $o$. Such a statement would be easy to apply at an optimum, but it would remove the orientation of every vector $x_i-o$ before the random rotation is analyzed. In particular, applying the fixed-axis factor $\sqrt2$ to every orientation and averaging would still give $\sqrt2$. The average cannot improve a bound whose right-hand side is already independent of the direction.

The estimate we need must therefore distinguish vectors of equal Euclidean length that point in different directions in the chosen axes. The coordinate $p$-norm does exactly this when $p<2$. A vector of Euclidean length one along an axis has coordinate $p$-norm one, whereas a vector along a diagonal has coordinate $p$-norm $2^{1/p-1/2}$. These two values become equal at $p=2$, which anticipates why the intermediate-exponent analysis has a different character from the quadratic endpoint.

There is a second requirement. After averaging, the directional contribution of agent $i$ must become a multiple of its $p$th-power distance from the benchmark. The angular average of a coordinate norm is proportional to Euclidean length, so it is natural to multiply that norm by $r_i^{p-1}$. The product has degree $p$ in the distance. Dividing the resulting sum by the benchmark cost to the power $p-1$ returns a quantity of degree one, as a social cost should be. This homogeneity consideration motivates the particular form of the next theorem.

\begin{theorem}[Reference-point inequality]\label{thm:reference}
Fix $1<p<2$ and a coordinate system. Let $m$ be a coordinate-wise median of $X$ and let $o\in\R^2$. Define
\[
 A=\sum_i\norm{x_i-o}_2^p,\qquad
 B=\sum_i\norm{x_i-o}_2^{p-1}\norm{x_i-o}_p.
\]
If $A>0$, then
\begin{equation}\label{eq:reference}
 C_p(X;m)\le 2^{1-1/p}\frac{B}{A^{(p-1)/p}}.
\end{equation}
\end{theorem}

All statements in this section use a fixed $p\in(1,2)$.

To understand the right-hand side of~\eqref{eq:reference}, consider a nonzero vector $x_i-o$ and write its contribution to $B$ as $r_i^p q(x_i-o)$. The ratio $B/A$ is then an average of the directional norm ratios, with weights proportional to the agents' $p$th-power distances from $o$. Distant agents receive more weight in this average because they also contribute more to the benchmark objective. These weights are analytical quantities, not voting weights assigned to the median rule. The median remains the ordinary unweighted coordinate median.

The arbitrary-reference-point formulation has a practical role in the proof. It allows us to translate $o$ to the origin without using first-order optimality conditions for $C_p$. We can then reflect individual axes to place the median in the nonnegative quadrant. The benchmark will be chosen to be an optimum only when the deterministic inequality is applied to the mechanism. Keeping these choices separate avoids making the auxiliary estimates depend on the particular equation satisfied by a minimizer of the social cost.
 Put
\[
 t=p-1,\qquad u=2-p=1-t,\qquad Q=2^{1/p-1/2}.
\]
For $x\in\R^2$, write $r(x)=\norm x_2$, $s(x)=\norm x_p$, and $q(x)=s(x)/r(x)$ when $x\ne0$. Norm comparison gives $1\le q(x)\le Q\le\sqrt2$.

The parameters $t$ and $u$ record two complementary features of the exponent. The derivative of a $p$th power involves $p-1=t$, while the change in a Euclidean distance power involves $p-2=-u$. Their sum is one, which will let two elementary bounds be interpolated multiplicatively in Lemma~\ref{lem:variation}. The upper limit $Q$ on $q(x)$ is the exact two-dimensional norm-comparison factor. In this section, $q(x)$ is a ratio of norms at a vector; it is not a new exponent or a change in the underlying distance.

The auxiliary scalar $\lambda$ will ultimately be chosen from the direction of the median itself. Until then, keeping it free allows us to prove a gradient estimate that works uniformly over the entire permitted interval. The following function combines the directional quantity $r^t s$ with the radial quantity $r^p$. The negative term is deliberate: after summation, it produces a one-variable expression whose maximum can be evaluated exactly.

For $\lambda\in[1,Q]$, define
\begin{equation}\label{eq:phi}
 \Phi_\lambda(x)=p\lambda^t r(x)^t s(x)-t\lambda^p r(x)^p,
 \qquad\Phi_\lambda(0)=0.
\end{equation}
This function is invariant under coordinate sign changes and is homogeneous of degree $p$. It is continuously differentiable away from the origin. At the origin its derivative is zero: the function is $O(r^p)$ and its gradient away from zero is $O(r^{p-1})$. Thus it is continuously differentiable everywhere.

There is a useful way to see where this function comes from. If $x\ne0$, then $\Phi_\lambda(x)/r(x)^p=p\lambda^{p-1}q(x)-(p-1)\lambda^p$. As a function of $q(x)$, this is the tangent line to $q^p$ at $q=\lambda$. Convexity therefore places it below $q(x)^p$. Multiplying by $r(x)^p$ says that $\Phi_\lambda(x)$ is a tangent-based surrogate for $s(x)^p$. When $\lambda=q(m)$, the surrogate is exact at the median: $\Phi_\lambda(m)=s(m)^p$.

This tangency explains both the opportunity and the difficulty. Replacing $s(x)^p$ by a smaller expression is useful if it leads to the desired mixed norm after summing, but the replacement introduces an error. The next two lemmas show that this error is controlled on the boundary of a median half-plane. They compare how quickly the norm ratio changes with a convexity gap in the coordinate that remains free. Without that comparison, the tangent surrogate would have no reason to satisfy the needed distance inequality.

We do not assume that $\Phi_\lambda$ is globally convex, or that a coordinate median minimizes its sum. Neither assertion is needed. The argument instead establishes three explicit facts about this function: a bound on a coordinate boundary, a lower bound on derivatives in nonnegative coordinate directions, and a coarse estimate valid on the whole plane. Those facts are then combined according to the position of a point relative to the median.

The first lemma controls how the ratio of the two norms changes along a coordinate line. Its constant need not be optimal; the stated estimate suffices to compare the two convexity remainders in Lemma~\ref{lem:boundary}.

\subsection{Variation of the norm ratio}

Fixing the first coordinate at $a>0$ leaves one nonnegative coordinate free. The quantity $q_v$ compares the coordinate $p$-norm of $(a,v)$ with its Euclidean length. It equals one near an axis and is largest near the diagonal. Consequently it is not monotone over the whole half-line in $v$. A direct global derivative bound would have to accommodate that change of direction. The proof below first uses the symmetry that exchanges the two coordinates, then works on the interval where the ratio is monotone.

The factors in~\eqref{eq:variation} are chosen for the next lemma rather than for their appearance in isolation. Squaring the estimate gives a bound on $r_v^p(q_v-q_b)^2$. This is exactly the form of the Taylor remainder introduced by the tangent surrogate. The right-hand side becomes a multiple of $(v+b)^{p-2}(v-b)^2$, which will also occur in a lower bound on the midpoint convexity gap of $v^p$. Thus the lemma aligns two different errors so that one can pay for the other.

\begin{lemma}\label{lem:variation}
Let $a>0$ and $v,b\ge0$. Write
\[
 r_v=(a^2+v^2)^{1/2},\qquad
 q_v=\frac{(a^p+v^p)^{1/p}}{r_v},
\]
and define $q_b$ in the same way. If $v+b>0$, then
\begin{equation}\label{eq:variation}
 r_v^{p/2}|q_v-q_b|
 \le\frac{|v-b|}{\sqrt2\,(v+b)^{u/2}}.
\end{equation}
\end{lemma}
\begin{proof}
Both sides scale by $a^{p/2}$ when $v,b$ are rescaled relative to $a$, so assume $a=1$.

More explicitly, write $v=a\widehat v$ and $b=a\widehat b$. The norm ratio does not change under this common scaling, while $r_v^{p/2}$ gains a factor $a^{p/2}$. The other side gains $a^{1-u/2}=a^{p/2}$. This verifies that setting $a=1$ loses no cases. We will restore the factor only at the end, so all intermediate trigonometric expressions involve ordinary dimensionless parameters.
 Define
\[
 q_p(z)=\frac{(1+z^p)^{1/p}}{\sqrt{1+z^2}},\qquad
 q_1(z)=\frac{1+z}{\sqrt{1+z^2}}.
\]
For $0<z\le1$, differentiation gives
\[
 q_p'(z)=\frac{z^{p-1}-z}
 {(1+z^2)^{3/2}(1+z^p)^{(p-1)/p}}.
\]
The convexity of $h\mapsto z^h$ yields $z^t\le1-t+tz$. Since $z^t\ge z$ and $(1+z^p)^{t/p}\ge1$, it follows that
\[
 0\le q_p'(z)\le u\frac{1-z}{(1+z^2)^{3/2}}=u q_1'(z).
\]
The inequality extends to $z=0$ by a one-sided limit. Both functions satisfy $q_j(z)=q_j(1/z)$ for $z>0$. Fold each argument into $[0,1]$ using $z\mapsto\min\{z,1/z\}$, with zero mapped to zero. Integrating the derivative inequality between the folded arguments gives
\begin{equation}\label{eq:contraction}
 |q_p(v)-q_p(b)|\le u|q_1(v)-q_1(b)|
 \quad(v,b\ge0).
\end{equation}

The comparison with $q_1$ is useful because $q_1$ has a simple trigonometric form. Folding does not replace a difference by a sum of two errors. Both functions take the same value at $z$ and $1/z$, so their endpoint differences are preserved exactly when each argument is folded. On the resulting interval inside $[0,1]$, the derivative inequality can be integrated directly between the two endpoints, in whichever order makes them increasing. This is why the absolute-value estimate remains valid even when $v$ and $b$ lie on different sides of one.

The claim is immediate if $v=b$. Otherwise put $v=\tan\theta$, $b=\tan\varphi$, where $\theta,\varphi\in[0,\pi/2)$, and let $M=(\theta+\varphi)/2$, $D=(\theta-\varphi)/2$. Using $q_1(\tan\theta)=\cos\theta+\sin\theta$ gives
\begin{equation}\label{eq:trig-difference}
 d:=\frac{|q_1(v)-q_1(b)|}{|v-b|}
 =\frac{|\cos M-\sin M|}{r_vr_b\cos D}.
\end{equation}

The tangent substitution is a geometric parametrization of a vector with first coordinate one. In particular, $r_v=1/\cos\theta$ and $r_b=1/\cos\varphi$. The midpoint angle $M$ measures the common orientation of the two vectors, and the half-difference $D$ measures their separation. The identity~\eqref{eq:trig-difference} follows by canceling the common factor $2\sin D$ in the two endpoint differences. When $v\ne b$, this factor is nonzero. The case of equal endpoints was separated precisely so that this division is legitimate.

One estimate on $d$ is not quite in the form required by~\eqref{eq:variation}. We therefore form two scaled versions, each with a simple bound. Their powers $u$ and $t$ will recover the exact exponents needed in the lemma. This step is an interpolation of two explicit scalar estimates, not an invocation of an abstract interpolation theorem.

Set $K_1=\sqrt{r_v(v+b)}\,d$ and $K_2=r_vd$. Since $v+b=\sin(2M)r_vr_b$,
\[
 K_1^2=\frac{\sin(2M)(1-\sin(2M))}{r_b\cos^2D}
 \le\frac{1/4}{1/2}=\frac12.
\]
Here $r_b\ge1$, $|D|\le\pi/4$, and $z(1-z)\le1/4$ for $z\in[0,1]$. Also $|D|\le M\le\pi/2-|D|$, so
\[
 |\cos M-\sin M|\le\cos|D|-\sin|D|\le\cos D.
\]
Equation~\eqref{eq:trig-difference} therefore implies $K_2\le1/r_b\le1$.

The bound on $K_1$ uses the product $\sin(2M)(1-\sin(2M))$, whose two factors cannot both be large. Its maximum is $1/4$, while the denominator is at least $1/2$. The estimate on $K_2$ instead uses the allowed interval for $M$ at a given $D$. On that interval, $|\cos M-\sin M|$ is largest at an endpoint, where it equals $\cos|D|-\sin|D|$. These are complementary estimates: the first carries the additional scale $v+b$, while the second supplies the missing power of $r_v$.

Multiplying the two estimates with exponents $u$ and $t$ is legitimate because both exponents are positive and sum to one. Explicitly, $K_1^uK_2^t=r_v^{u/2+t}(v+b)^{u/2}d^{u+t}$. Thus it equals $r_v^{p/2}(v+b)^{u/2}d$, exactly the scale left after applying the contraction estimate. The remaining factor $u$ is retained from that contraction. Dropping it would give $2^{-u/2}$, which is too large to supply the uniform coefficient required for the boundary comparison.

As $u+t=1$ and $u/2+t=p/2$, equation~\eqref{eq:contraction} now yields
\[
 r_v^{p/2}(v+b)^{u/2}
 \frac{|q_p(v)-q_p(b)|}{|v-b|}
 \le uK_1^uK_2^t\le u2^{-u/2}\le2^{-1/2}.
\]
The last inequality follows because $h(u)=u2^{-u/2}$ is increasing on $(0,1]$: its derivative is $2^{-u/2}(1-u\log2/2)>0$. Restoring the scale proves~\eqref{eq:variation}.
\end{proof}

The conclusion of Lemma~\ref{lem:variation} is asymmetric in $v$ and $b$ because the Taylor error in the next lemma is multiplied by the Euclidean length at the evaluation point $v$. No symmetry of that prefactor is needed. The estimate is also deliberately stronger than a mere continuity statement: it has the precise quadratic scale after squaring. At $v=b$ both the norm-ratio error and the compensating midpoint gap vanish, so preserving their rates of vanishing is essential.

The proof exposes where the restriction $p<2$ first enters quantitatively. It supplies the positive factor $u=2-p$ in the contraction toward $q_1$. As $p$ approaches two, the ratio of the coordinate $p$-norm to Euclidean length becomes nearly constant, and its variation correspondingly shrinks. The next lemma translates this directional regularity into a bound involving actual distances to the median.

\subsection{Boundary and gradient estimates}

Suppose the median has nonnegative coordinates $(a,b)$ and a point has first coordinate exactly $a$. Its displacement from the median is then entirely in the second coordinate. We need an estimate at this boundary before moving the point farther to the right. The parameter $\lambda=q(a,b)$ makes the auxiliary function exact at the median, so the estimate should also be exact when the free coordinate equals $b$.

The term $a^p-b^p$ in the next statement may appear unusual if it is viewed only as a bound on a function. It is dictated by the later indicator correction. Crossing the first median threshold will remove $2a^p$ from the otherwise positive constant $a^p+b^p$. Rearranging the desired pointwise estimate then requires exactly $a^p-b^p$ on the right. The boundary lemma is designed to produce that term, together with enough room to bound the distance in the second coordinate.

\begin{lemma}[Boundary estimate]\label{lem:boundary}
For $a>0$, $b,v\ge0$, and $\lambda=q(a,b)$,
\begin{equation}\label{eq:boundary}
 \Phi_\lambda(a,v)\ge a^p-b^p+2^{1-p}(v+b)^p.
\end{equation}
\end{lemma}
\begin{proof}
Use $r_v,q_v,q_b$ from Lemma~\ref{lem:variation}. Since $\lambda=q_b$,
\begin{align*}
 D&:=(a^p+v^p)-\Phi_\lambda(a,v)\\
  &=r_v^p\bigl[q_v^p-q_b^p-pq_b^{p-1}(q_v-q_b)\bigr].
\end{align*}
For $f(z)=z^p$ on $z\ge1$, $0<f''(z)\le p(p-1)$. Taylor's theorem, applied between $q_v,q_b\ge1$, and Lemma~\ref{lem:variation} show that
\begin{equation}\label{eq:D}
 0\le D\le\frac{p(p-1)}2r_v^p(q_v-q_b)^2
 \le\frac{p(p-1)}4(v+b)^{p-2}(v-b)^2.
\end{equation}

Here $D$ is the price of replacing the coordinate $p$th power by its tangent surrogate. Convexity ensures that this price is nonnegative. To bound it from above, we use an upper bound on the second derivative over the interval between the two norm ratios. That interval lies above one, where the negative exponent $p-2$ makes $z^{p-2}$ at most one. The variation lemma then converts the error from a difference of norm ratios into a difference of the original coordinates.

We next look for a nonnegative quantity that is at least as large as $D$. The midpoint gap $J$ is natural because the desired distance term involves the sum $v+b$. Rewriting $2^{1-p}(v+b)^p$ as twice the $p$th power of the midpoint makes $J$ the usual two-point convexity gap. Unlike the preceding Taylor estimate, this part of the argument needs a lower bound on curvature. Both estimates are valid because they apply on different intervals: the norm-ratio interval is bounded below by one, whereas the coordinate interval is bounded above by $v+b$.

On the other hand, let $J=v^p+b^p-2^{1-p}(v+b)^p$. If $v,b>0$, then throughout the interval with endpoints $v,b$,
\[
 f''(z)\ge\kappa:=p(p-1)(v+b)^{p-2}.
\]
The midpoint inequality for the convex function $f(z)-\kappa z^2/2$ gives
\begin{equation}\label{eq:J}
 J\ge\frac{\kappa}{4}(v-b)^2.
\end{equation}

The factor $1/4$ can be checked without recalling a specialized strong-convexity formula. Put $w=(v+b)/2$ and $g(z)=f(z)-\kappa z^2/2$. Convexity gives $g(v)+g(b)-2g(w)\ge0$. On returning the quadratic terms to the other side, their contribution is $\kappa(v^2+b^2-2w^2)/2=\kappa(v-b)^2/4$. This is exactly the coefficient in~\eqref{eq:J}. Its agreement with the final coefficient in~\eqref{eq:D} is the reason the comparison closes without a residual error.

The curvature lower bound uses the largest convenient coordinate scale, $v+b$, rather than the smaller endpoint. Since $p-2$ is negative, $z\le v+b$ gives $z^{p-2}\ge(v+b)^{p-2}$. The curvature may become unbounded near zero, but this does not invalidate a lower bound. We first apply the argument on an interval with positive endpoints, where differentiation is ordinary, and then use continuity of the original expressions to include a zero endpoint. No finite value of $f''(0)$ is asserted.

When exactly one endpoint is zero, take a limit from positive endpoints in~\eqref{eq:D}--\eqref{eq:J}. If $v=b=0$, equation~\eqref{eq:boundary} holds directly with $\lambda=1$. In every case $D\le J$, and hence
\[
 \Phi_\lambda(a,v)\ge a^p+v^p-J
 =a^p-b^p+2^{1-p}(v+b)^p.
\qedhere
\]
\end{proof}

The boundary estimate has two useful checks. If $v=b$, the midpoint gap and tangent error both vanish and the displayed inequality is an equality. If $v=b=0$, the norm ratio is one and the auxiliary expression reduces directly to $a^p$. These checks do not prove the general statement, but they show why neither a division by $v-b$ nor a finite second derivative at zero should appear in its final formulation. The proof uses limits only at the coordinate endpoints where such care is necessary.

We now need to extend information from a boundary line into an entire region. For that purpose, a lower bound on the coordinate derivatives of $\Phi_\lambda$ is more useful than a new bound on its values. The estimate below compares the derivative with that of the separable function $2^{1-p}(x_1^p+x_2^p)$ in nonnegative coordinate directions. When integrated along a path, it will also control the region where both coordinates lie beyond the median.

\begin{lemma}[Coordinate gradient bound]\label{lem:gradient}
For $\lambda\in[1,Q]$ and $x_j\ge0$,
\begin{equation}\label{eq:gradient}
 \partial_j\Phi_\lambda(x)\ge p2^{1-p}x_j^{p-1}.
\end{equation}
\end{lemma}
\begin{proof}
Suppose $x_j>0$, and write $z=x_j/r(x)\in(0,1]$, $q=q(x)$. Differentiating~\eqref{eq:phi} gives
\begin{equation}\label{eq:gradient-formula}
 \frac{\partial_j\Phi_\lambda(x)}{p x_j^t}
 =\left(\frac\lambda q\right)^t
   +t\lambda^t(q-\lambda)z^{1-t}.
\end{equation}

The variable $z$ records the fraction of Euclidean length carried by coordinate $j$. It is between zero and one, irrespective of the sign of the other coordinate. The first term in~\eqref{eq:gradient-formula} is positive. The second can have either sign, so a single estimate that discards it would be invalid. Splitting according to whether $q$ is larger or smaller than $\lambda$ handles precisely this issue. In the second case we must bound how negative the correction can become, which is why the restriction $\lambda\le Q$ matters.

The normalization by $p x_j^t$ also clarifies the target: we need the remaining dimensionless expression to be at least $2^{-t}$. This removes the scale of $x$ from the calculation. At a zero coordinate the normalization is unavailable, but the original derivative inequality still makes sense and follows by continuity. The proof therefore treats positive coordinates algebraically and returns to the zero-coordinate boundary only after the estimate has been established.

If $q\ge\lambda$, the second term is nonnegative and the first is at least $Q^{-t}\ge2^{-t}$. If $q<\lambda$, the right-hand side is at least $1-t\lambda^t(\lambda-1)$, since $q\ge1$ and $z^{1-t}\le1$. Furthermore,
\[
 \lambda^t(\lambda-1)=\lambda^p-\lambda^t
 \le Q^p-1=2^{1-p/2}-1\le\sqrt2-1<\tfrac12.
\]
Thus~\eqref{eq:gradient-formula} is at least $1-t/2\ge2^{-t}$. The final inequality follows from the chord bound for the convex function $t\mapsto2^{-t}$ on $[0,1]$. Continuous differentiability handles $x_j=0$.

The two sign cases in this proof produce the same derivative threshold for different reasons. When $q\ge\lambda$, there is no negative correction, so a coarse positive lower bound on the first term suffices. When $q<\lambda$, the first term is at least one, and some of that extra margin can absorb the negative correction. The chord inequality $2^{-t}\le1-t/2$ converts this remaining linear margin into the uniform threshold $2^{-t}$. The proof does not establish equality throughout the range, and it does not need to: this derivative bound is an ingredient for extending a boundary inequality, rather than a separate claim about the sharpest possible derivative constant.

\end{proof}

The gradient bound does not require $\lambda=q(x)$, and this uniformity is important. In the later pointwise argument, $\lambda$ is fixed by the median, while $x$ ranges over all input points. The two norm ratios can differ in either direction. The proof of Lemma~\ref{lem:gradient} explicitly covers both possibilities and therefore permits integration while $q(x)$ varies along a coordinate path.

The remaining ingredient is a bound valid without any restriction on which side of a median threshold a point lies. A triangle inequality will supply it, but the coefficients must be checked carefully. The auxiliary function contains a negative term, so its nonnegativity and its radial lower bound cannot simply be assumed. The next lemma identifies a positive radial coefficient and shows that the weighted H\"older factor does not exceed the desired constant.

\begin{lemma}[Global estimate]\label{lem:global}
If $m\ne0$ and $\lambda=q(m)$, then for every $x\in\R^2$,
\begin{equation}\label{eq:global}
 \norm{x-m}_2^p\le2^{p-1}\bigl(\Phi_\lambda(x)+s(m)^p\bigr).
\end{equation}
\end{lemma}
\begin{proof}
Since $s(x)\ge r(x)$, we have $\Phi_\lambda(x)\ge c r(x)^p$, where
\[
 c=\lambda^t(p-t\lambda)
   =\lambda^t(1-t\delta)>0,\qquad
 \delta=\lambda-1\in[0,\sqrt2-1].
\]

The scalar $c$ is a uniform radial lower bound for the auxiliary function. The identity $p-t\lambda=1-t\delta$ makes positivity transparent: both $t$ and $\delta$ are strictly less than one. Its exact value is not the final approximation factor. Instead, it will be paired with the coefficient $\lambda^p$ already present in $s(m)^p=\lambda^p r(m)^p$. Weighted H\"older measures the cost of bounding the sum of the two Euclidean lengths using those two coefficients.

For ordinary equal coefficients, the familiar inequality is $(r+s)^p\le2^{p-1}(r^p+s^p)$. Here one coefficient may be smaller than one and the other larger than one. Applying the equal-coefficient inequality separately would lose that balance. The calculation below shows that their inverse powers still sum to at most two, so exactly the same factor $2^{p-1}$ remains available.

Concavity gives $(1-\delta)^t\le1-t\delta$, so
\[
 c^{-1/t}=\lambda^{-1}(1-t\delta)^{-1/t}
 \le\frac1{1-\delta^2}.
\]
Also $p/t=1+1/t\ge2$, and consequently
\begin{equation}\label{eq:holder-coefficient}
 c^{-1/t}+\lambda^{-p/t}
 \le\frac1{1-\delta^2}+\frac1{(1+\delta)^2}
 =\frac2{(1-\delta)(1+\delta)^2}\le2.
\end{equation}
Indeed, $\delta+\delta^2\le2-\sqrt2<1$, which implies
$(1-\delta)(1+\delta)^2=1+\delta(1-\delta-\delta^2)\ge1$.

Here the last rational comparison has a simple interpretation. The denominator $(1-\delta)(1+\delta)^2$ measures whether the gain from the coefficient greater than one compensates for the loss from the coefficient below one. On the permitted range of $\delta$, the expansion $1+\delta(1-\delta-\delta^2)$ is at least one, so that compensation is sufficient. The upper limit $\sqrt2-1$ is used only to certify this sign; the proof does not require the actual median's norm ratio to reach that limit.

To spell out the weighted H\"older step, use conjugate exponents $p$ and $p/t$, whose reciprocal sum is one because $t=p-1$. Write the sum of lengths as the inner product of the two vectors
\[
 \bigl(c^{1/p}r(x),\lambda r(m)\bigr)
 \quad\hbox{and}\quad
 \bigl(c^{-1/p},\lambda^{-1}\bigr).
\]
The $p$th power of the first vector's $\ell_p$ norm is $cr(x)^p+\lambda^p r(m)^p$. The $p$th power of the second vector's $\ell_{p/t}$ norm is $(c^{-1/t}+\lambda^{-p/t})^t$. Their product is the expression used below. This derivation also verifies the exponents on both coefficients, which cannot be replaced by the coefficients themselves. Both are positive, so the application is valid even if one of the lengths is zero.

By the triangle inequality and weighted H\"older inequality,
\begin{align*}
 \norm{x-m}_2^p
 &\le(r(x)+r(m))^p\\
 &\le(c^{-1/t}+\lambda^{-p/t})^t
       \bigl(cr(x)^p+\lambda^pr(m)^p\bigr)\\
 &\le2^t\bigl(\Phi_\lambda(x)+s(m)^p\bigr).
\end{align*}
For clarity, the middle inequality is H\"older with conjugate exponents $p$ and $p/t$, applied to $(c^{1/p}r(x),\lambda r(m))$ and $(c^{-1/p},\lambda^{-1})$.
\end{proof}

Lemma~\ref{lem:global} does not use any median condition. It is an inequality between two arbitrary vectors, with $\lambda$ chosen from one of them. Its role is limited but necessary: it handles points lying below both of the median's coordinate thresholds after the coordinate reflections. The triangle inequality is sufficient in that region because no negative indicator correction has to be paid. In the other regions, the boundary and derivative estimates provide the stronger information that the global bound lacks.

This separation of roles prevents a misleading interpretation of the proof. We are not expecting one triangle inequality to be sharp on every point of a worst-case profile. The argument uses different pointwise estimates in different regions and then exploits counting constraints after summation. Individual intermediate inequalities can have slack without preventing a sharp final constant. Sharpness is established globally by the lower-bound family, not by requiring simultaneous equality in every auxiliary estimate for one finite input.

\subsection{The pointwise inequality and median cancellation}

The next estimate is organized so that the positive constants are canceled by the median constraints. The reference point has already been translated to the origin; the point $m$ in the lemma need not itself be a median.

To see the form of the correction, start with the global estimate, whose right-hand side contains $a^p+b^p$. If we summed that estimate over all agents, we would acquire an unwanted term $n(a^p+b^p)$. The median constraints suggest exactly how to remove it. At least $n/2$ points satisfy $x_1\ge a$, so subtracting $2a^p$ on that half-plane cancels the full $na^p$ after summation. A corresponding subtraction on $x_2\ge b$ cancels $nb^p$.

The challenge is to justify these subtractions pointwise. Points in the intersection of the two half-planes must pay both corrections; points in neither half-plane pay neither. This gives the four regions in the proof. The numerical coefficient two in each correction comes from the one-half count in the definition of a median. It is not a free parameter tuned after observing the profile.

The indicator functions are discontinuous across their threshold lines, but the inequality is still a statement about ordinary continuous distances and the auxiliary function. We choose the value one at equality because the median counts include equality. The boundary lemma must therefore be strong enough to support the subtraction already on the threshold, not merely at points strictly beyond it. That is why it was proved before the region-by-region argument.

\begin{lemma}[Half-plane correction]\label{lem:pointwise}
Let $m=(a,b)$ with $a,b>0$, and set $\lambda=q(m)$. For every $x=(x_1,x_2)$,
\begin{equation}\label{eq:pointwise}
 \begin{split}
 \norm{x-m}_2^p\le2^{p-1}\bigl[\Phi_\lambda(x)+a^p+b^p
 &-2a^p\ind_{\{x_1\ge a\}}\\
 &-2b^p\ind_{\{x_2\ge b\}}\bigr].
 \end{split}
\end{equation}
The same inequality extends to $a,b\ge0$ with $m\ne0$.
\end{lemma}
\begin{proof}
First assume $a,b>0$. There are four regions.

\emph{Both coordinates below the thresholds.} If $x_1<a$ and $x_2<b$, both indicators vanish and Lemma~\ref{lem:global} gives~\eqref{eq:pointwise}.

No sign condition on $x_1$ or $x_2$ is imposed in this first region. Either coordinate can be very negative, so an argument based only on integrating positive-coordinate derivatives would not cover it. The global estimate was formulated for the entire plane precisely to avoid such an omission. Its reference term $a^p+b^p$ remains available in full because neither correction is active.

\emph{Only the first coordinate at or above its threshold.} Suppose $x_1\ge a$ and $x_2<b$. At $x_1=a$, sign invariance and Lemma~\ref{lem:boundary} imply
\begin{align}
 \Phi_\lambda(a,x_2)
 &=\Phi_\lambda(a,|x_2|)\notag\\
 &\ge a^p-b^p+2^{1-p}(|x_2|+b)^p\notag\\
 &\ge a^p-b^p+2^{1-p}|x_2-b|^p.
 \label{eq:boundary-extension}
\end{align}

The boundary calculation uses $|x_2|+b$ rather than $x_2+b$ so that it remains valid even if the second coordinate is negative. The triangle inequality on the real line then bounds $|x_2-b|$. At this stage we have established the required lower bound for the difference between the auxiliary function and a scaled distance power only at $x_1=a$. To extend it, we hold the second coordinate fixed and show that this difference cannot decrease as the first coordinate moves to the right.

Fix $x_2$ and define $F(z)=\Phi_\lambda(z,x_2)-2^{1-p}\norm{(z,x_2)-m}_2^p$ for $z\ge a$. If $z>a$, write $d=\norm{(z,x_2)-m}_2\ge z-a$. Because $p-2<0$, $(z-a)d^{p-2}\le(z-a)^{p-1}$. Lemma~\ref{lem:gradient} gives
\[
 F'(z)\ge p2^{1-p}\bigl[z^{p-1}-(z-a)^{p-1}\bigr]\ge0.
\]
At $z=a$ use continuity. Thus~\eqref{eq:boundary-extension} yields
$\Phi_\lambda(x)\ge a^p-b^p+2^{1-p}\norm{x-m}_2^p$, which is exactly~\eqref{eq:pointwise} in this region.

The negative exponent in the derivative of the distance is helpful here. Since the full distance is at least its first-coordinate displacement, raising it to $p-2<0$ reverses the comparison. The increase in the distance power is therefore no larger than the separable derivative $p(z-a)^{p-1}$. Lemma~\ref{lem:gradient} supplies enough growth in the auxiliary function to cover this increase. Keeping track of that reversal is essential; with the inequality in the other direction, monotonicity of $F$ would not follow.

\emph{Only the second coordinate at or above its threshold.} The case $x_1<a$, $x_2\ge b$ follows by exchanging the coordinates.

\emph{Both coordinates at or above the thresholds.} If $x_1\ge a$, $x_2\ge b$, then $\Phi_\lambda(m)=s(m)^p=a^p+b^p$. Integrate Lemma~\ref{lem:gradient} along the coordinate path from $(a,b)$ to $(x_1,b)$ to $(x_1,x_2)$. This gives
\begin{align*}
 \Phi_\lambda(x)-a^p-b^p
 &\ge2^{1-p}(x_1^p-a^p+x_2^p-b^p)\\
 &\ge2^{1-p}\bigl((x_1-a)^p+(x_2-b)^p\bigr)\\
 &\ge2^{1-p}\norm{x-m}_2^p.
\end{align*}

The first integrated estimate follows because an antiderivative of $p2^{1-p}z^{p-1}$ is $2^{1-p}z^p$. Along the horizontal segment the second coordinate stays equal to $b$, and along the vertical segment the first stays equal to $x_1$; the gradient lemma applies on both segments with the same fixed $\lambda$. To compare a difference of powers with a displacement power, write $x_1=a+(x_1-a)$ and use nonnegativity of both summands. The same argument applies to the second coordinate. This step would fail if a coordinate were below its threshold, which explains why this integration argument belongs specifically to the fourth region.

For the last norm comparison, set $d_1=x_1-a$ and $d_2=x_2-b$. The inequality $\norm{(d_1,d_2)}_p\ge\norm{(d_1,d_2)}_2$ for $p<2$, raised to power $p$, says exactly that $d_1^p+d_2^p$ dominates the Euclidean distance power. Thus the separable increments obtained from integration are sufficient to cover the nonseparable objective. We use the coordinate norm only for this comparison; the actual distance in the conclusion remains Euclidean.

The second step uses $(v+w)^p\ge v^p+w^p$ for nonnegative $v,w$; the last uses $p<2$. This establishes the fourth case.

To extend the result to $m=(0,b)$ with $b>0$, apply the proved inequality to $m_\varepsilon=(\varepsilon,b)$ and let $\varepsilon\downarrow0$. The first indicator term has absolute value at most $2\varepsilon^p$, the second is unchanged, and all other terms converge continuously. The case $m=(a,0)$ is symmetric. These are pointwise limits; $m_\varepsilon$ is not required to be a median.
\end{proof}

Two boundary issues are resolved by the formulation of Lemma~\ref{lem:pointwise}. First, points may lie exactly on one or both threshold lines; the corresponding cases already include equality. Second, the median may have a zero coordinate after the reference point has been translated to zero. The limiting argument changes the comparison point in the pointwise lemma, not the input profile. It does not assert that the perturbed comparison point is still a median. Only after passing to the limit do we apply the actual median count.

We now return to the finite profile. Up to this point, the lemmas have been inequalities for one vector or two vectors. The only step that uses the collective structure of the agents is the summation of the half-plane corrections. This is the bridge from geometric analysis to the order-statistic property of the mechanism. It is also the reason the proof does not depend on a detailed classification of median configurations.

\begin{proof}[Proof of Theorem~\ref{thm:reference}]
Translate $o$ to the origin and reflect coordinate axes if necessary so that $m=(a,b)$ has $a,b\ge0$. These transformations preserve both norms and the median property. If $m=0$, then $C_p(X;m)=A^{1/p}$ and $B\ge A$, which proves~\eqref{eq:reference}. Otherwise set $\lambda=q(m)$ and sum~\eqref{eq:pointwise} over $x_1,\ldots,x_n$. By~\eqref{eq:median-count},
\[
 \#\{i:x_{i1}\ge a\}\ge n/2,\qquad
 \#\{i:x_{i2}\ge b\}\ge n/2.
\]
The total contribution of the constant and indicator terms is therefore nonpositive.

For the first coordinate, that contribution is $a^p(n-2N_1)$, where $N_1=\#\{i:x_{i1}\ge a\}$. Since $N_1\ge n/2$, it is at most zero. The second coordinate gives the identical expression with $b$ and $N_2$. Overlap between the two half-planes causes no problem: a point in their intersection legitimately receives both corrections, exactly as established in the fourth region of the pointwise lemma. Repeated coordinates likewise cause no difficulty because the counts include all agents and include equality.
 Hence
\begin{equation}\label{eq:sum-phi}
 C_p(X;m)^p\le2^{p-1}\sum_i\Phi_\lambda(x_i)
 =2^{p-1}\bigl(p\lambda^{p-1}B-(p-1)\lambda^pA\bigr).
\end{equation}
For $\lambda>0$, the derivative of $h(\lambda)=p\lambda^{p-1}B-(p-1)\lambda^pA$ is
\[
 h'(\lambda)=p(p-1)\lambda^{p-2}(B-\lambda A).
\]

The derivative changes sign at $B/A$: it is positive below that value and negative above it. Thus the required bound is a direct one-variable maximization, not an assumption that the geometrically chosen $\lambda=q(m)$ happens to equal $B/A$. In general those quantities differ. We are simply bounding the value at the actual parameter by the largest value allowed on the positive half-line. Because $A>0$ and $B\ge A$, the maximizing point is well defined and positive.

Its global maximum is attained at $\lambda=B/A$, with value $B^p/A^{p-1}$. Applying this upper bound in~\eqref{eq:sum-phi} and taking $p$th roots proves the theorem.
\end{proof}

After this maximization, the dependence on the median's own norm ratio has disappeared. The remaining right-hand side uses only the reference point and the profile, and it is linear in $B$ after taking a $p$th root. This is the structural feature needed for rotation. A looser estimate such as replacing every $q(x_i)$ by its maximum $Q$ would recover a direction-free factor but destroy the angular average that the next section exploits.

The proof also explains the apparent mixture of geometric and combinatorial ingredients. The difficult inequalities justify a certificate for each point, while the median supplies only two counts. There is no assumed distribution of directions and no assumed symmetry in the upper-bound profile. The symmetry used later to demonstrate sharpness is a property of a chosen lower-bound family, not a hidden restriction on this theorem.

\begin{remark}[Sharpness of the deterministic coefficient]
The coefficient $2^{1-1/p}$ in Theorem~\ref{thm:reference} cannot be reduced. Take $x_1=(-1,0)$, $x_2=(1,0)$, $o=0$, and the lower coordinate median $m=(-1,0)$. Then $A=B=2$ and $C_p(X;m)=2$, giving equality in~\eqref{eq:reference}. This observation concerns the deterministic inequality; the rotated mechanism requires the separate tightness argument in Section~\ref{sec:lower}.
\end{remark}

\section{Averaging the Reference-point Inequality}\label{sec:rotation}

We now pass from the deterministic estimate to the randomized mechanism. No optimization over the angle-dependent facility location is needed.

The order of operations matters. We first fix the input profile and a single optimal reference point in the physical plane. We then express those same points in different coordinate systems. The mechanism's median changes with the coordinate system, but the reference point does not move in the physical plane, and its Euclidean distances to the agents remain constant. Theorem~\ref{thm:reference} is applicable separately in every coordinate system because it holds for every reference point. We can therefore average its right-hand side without needing an explicit formula for the median trajectory.

There is a useful distinction between rotating the coordinate axes and changing the instance. In the calculation below, the agents retain their physical positions. Their two numerical coordinates change because the basis changes. A fixed Euclidean vector consequently has constant length but a varying coordinate $\ell_p$ norm. The average of that latter norm is the elementary geometric quantity isolated by the next lemma. This is the only angle-dependent expression that remains after the deterministic theorem has been applied.

\begin{lemma}[Rotational average]\label{lem:average}
For a fixed $z\in\R^2$ and uniform $\Theta\in[0,2\pi)$,
\begin{equation}\label{eq:average}
 \E_\Theta\left( |\ip z{e_\Theta}|^p+|\ip z{f_\Theta}|^p\right)^{1/p}
 =\norm z_2\,\frac2\pi\int_0^{\pi/2}
 (\cos^p\theta+\sin^p\theta)^{1/p}\dd\theta.
\end{equation}
\end{lemma}
\begin{proof}
For $z=0$ the identity is immediate. Otherwise write $z=\norm z_2(\cos\varphi,\sin\varphi)$. Its rotated coordinates have absolute values $\norm z_2|\cos(\Theta-\varphi)|$ and $\norm z_2|\sin(\Theta-\varphi)|$. A shift of a uniform angle preserves its distribution modulo $2\pi$. The integrand has period $\pi/2$, giving~\eqref{eq:average}.
\end{proof}

The identity in Lemma~\ref{lem:average} is independent of the direction of $z$. This does not mean that the coordinate norm is the same in every direction: for $1<p<2$, a vector pointing diagonally has a larger coordinate $\ell_p$ norm than an axis-aligned vector of the same Euclidean length. Rather, a full uniform rotation makes every initial direction encounter the same set of relative angles with the same weights. The factor $2/\pi$ normalizes integration over one quadrant, whose length is $\pi/2$.

Another point is that the lemma averages a norm itself. It does not first average its $p$th power and then take a root. The latter procedure would replace the mean of a concave function by that function of the mean, producing the larger moment bound discussed in Section~\ref{sec:comparison}. The strengthened deterministic theorem was designed precisely so that the norm in this lemma occurs linearly and can be averaged in its present form.

\begin{proof}[Upper bound in Theorem~\ref{thm:main}]
Fix a profile with positive optimum. Choose a minimizer $o$ and put $r_i=\norm{x_i-o}_2$ and $R_p=(\sum_i r_i^p)^{1/p}=\OPT_p(X)$. In every rotated coordinate system, Theorem~\ref{thm:reference} gives
\[
 C_p(X;m_\theta)\le
 \frac{2^{1-1/p}}{R_p^{p-1}}
 \sum_i r_i^{p-1}
 \left(|\ip{x_i-o}{e_\theta}|^p+|\ip{x_i-o}{f_\theta}|^p\right)^{1/p}.
\]
Taking expectations, exchanging the finite sum with the expectation, and using Lemma~\ref{lem:average}, we obtain
\begin{align*}
 \E C_p(X;m_\Theta)
 &\le\frac{2^{2-1/p}}{\pi R_p^{p-1}}
 \left[\int_0^{\pi/2}(\cos^p\theta+\sin^p\theta)^{1/p}\dd\theta\right]
 \sum_i r_i^p\\
 &=L(p)\OPT_p(X).
\end{align*}
The deterministic reference-point inequality uses only the median inequalities~\eqref{eq:median-count}. Thus the same upper bound applies to every measurable selection within the median intervals.
\end{proof}

The factor $r_i^{p-1}$ in the reference-point inequality is essential to this averaging step: the rotation average contributes another factor $r_i$, and the resulting sum is precisely the optimal $p$th-power cost. Although the intermediate coordinate norm depends on the angle, its coefficient and the denominator do not.

To see the cancellation explicitly, the numerator after averaging contains $\sum_i r_i^{p-1}r_i=R_p^p$, whereas the denominator is $R_p^{p-1}$. Their quotient is $R_p$, which is the benchmark cost itself. Thus no estimate involving the number of agents, the diameter of the instance, or the distribution of radial distances is needed at this stage. All of those features were already accommodated by the deterministic inequality. Agents located exactly at the reference point have zero weight and contribute zero, so no division by an individual distance is necessary.

The expectation and summation are interchangeable because the sum has finitely many terms, each bounded by $Q r_i$ before multiplication by its fixed weight. There is no limiting profile in the upper-bound proof. Similarly, the reference point is chosen before the random angle is drawn; an interchange of minimization and expectation is neither asserted nor required. The cost of the rotated median is bounded relative to the same optimum at every angle, which is stronger than bounding it relative to an angle-dependent comparator.

The median-selection statement should be read at the level of approximation. When a coordinate has a nontrivial median interval, every point of that interval satisfies the two half-sample inequalities used by the proof. Consequently, choosing different points in these intervals cannot invalidate the upper bound, provided the choice is measurable so that the expectation is defined. It does not follow that every such selection rule is strategyproof: the incentive claim in Section~\ref{sec:model} was for the specified order-statistic rule. Keeping these two claims separate avoids imposing an unnecessary tie-breaking restriction on the geometric theorem.

\subsection{A worked application and the slack in the bound}

The upper-bound argument applies without knowing the median as a function of the angle. A symmetric example in which that function is known helps distinguish what the theorem proves from what its intermediate expressions mean. Consider five agents at
\[
 (1,0),\quad(-1,0),\quad(0,1),\quad(0,-1),\quad(0,0).
\]
For every axis, their scalar projections occur in two opposite pairs together with zero. The middle of the five projections is therefore zero, including when some projections coincide. Both coordinate medians are zero for every orientation, and the mechanism always places the facility at the origin. In this example the randomization changes the coordinates used internally but not the final facility.

The origin is also optimal for the social objective. Indeed, the sum of the five distance powers is a convex even function of the facility location: reflecting the facility through the origin permutes the four nonzero agents and leaves the fifth fixed. Its value at the midpoint of $y$ and $-y$ is at most their common value. Thus the optimal $p$th-power cost is four, the cost is $4^{1/p}$, and the actual approximation ratio of this profile is one. This conclusion is valid for each $p$ in the parameter range, without a numerical optimization.

We can nevertheless apply the reference-point theorem to this profile and examine its right-hand side. Choose $o=0$. Four of the radial distances $r_i$ equal one and the fifth equals zero, so $A=4$. In axes at angle $\theta$, each of the four nonzero displacement vectors has the same coordinate norm
\[
 h_p(\theta)=\bigl(|\cos\theta|^p+|\sin\theta|^p\bigr)^{1/p}.
\]
Changing from a horizontal vector to a vertical one exchanges the absolute coordinate values, and changing to its opposite only changes signs. Consequently, $B=4h_p(\theta)$. Substitution in Theorem~\ref{thm:reference} gives
\[
 C_p(X;m_\theta)\le
 2^{1-1/p}\,4^{1/p}h_p(\theta).
\]
After division by the actual optimum $4^{1/p}$, the resulting bound is $2^{1-1/p}h_p(\theta)$. Its angular average is $L(p)$, even though the actual ratio is one. There is no contradiction: a universal sharp bound need not be an equality on every profile. This example makes visible the slack that the proof is allowed to have away from its worst-case family.

The example also separates three kinds of angle dependence. The Euclidean benchmark distances are fixed. The coordinate norms vary with the basis. The mechanism's output may or may not vary; here it does not. Our averaging argument requires only the first two facts and makes no assumption about the third. This is why it can be applied to irregular profiles for which a direct piecewise formula for the median would be cumbersome. The resulting bound remains correct even when the directional expression overestimates an unchanging output.

Why not improve the universal constant by exploiting that slack? Doing so would require an additional inequality that holds on every profile, including the asymmetric family in the next section. Symmetry supplies such information in the present example, but symmetry is not part of the model. The lower-bound family is designed so that the mechanism's output does vary with the angle in a way that realizes the same angular expression asymptotically. Hence the slack in this symmetric instance cannot be subtracted uniformly.

This calculation is also a practical guide for applying the theorem. One first chooses the comparator, computes Euclidean radii, and forms the weighted sum of coordinate norms. One then applies the deterministic inequality, retaining that sum in its mixed form. Only after this step does one average the coordinate norms and identify the benchmark cost. Replacing $h_p(\theta)$ by its maximum $Q$ before averaging would produce the larger factor $2^{1-1/p}Q=\sqrt2$, regardless of the fact that most orientations have a smaller coordinate norm. The same unnecessary loss would occur on an arbitrary profile if every directional term were bounded by its maximum separately.

Finally, the example illustrates why it is useful to keep the optimal reference point distinct from the median in the statement of the deterministic theorem, even when they happen to coincide here. On a general profile, the median is selected by ranks and the optimum by minimization of a distance objective. There is no reason for the two points to agree. The theorem relates their costs without asserting any shared first-order optimality condition. Its proof uses the median counts for one point and only the distances to the other.

\section{Matching Lower Bound}\label{sec:lower}

We give a complete lower-bound proof using the two-cluster-and-outlier geometry of \citet{barak2026} and \citet{chan2026}. The integral lower bound is already established by the latter work. The construction below is included to exhibit the finite-instance meaning of tightness and to specify the limiting argument.

An upper bound must cover every profile, whereas a matching lower bound needs only a carefully chosen family of profiles. The family will have two equally large clusters, separated by a fixed distance, and one increasingly distant outlier. The cluster multiplicities grow faster than the outlier's distance cost. This makes the outlier negligible for the optimal normalized social cost, while allowing it to determine the coordinate median through its position in the sorted order. The mechanism reacts to ranks rather than to the magnitudes of all distances, and that distinction is what the construction uses.

There are three parts to the argument. We first normalize the optimal cost without solving exactly for the optimal facility. We then calculate the coordinate medians, which is possible because all but one agent belong to two repeated locations. Finally, we average the limiting costs and justify passing the limit through the integral. These steps serve different purposes: symmetry controls the denominator, clipping controls the output, and a uniform bound controls the expectation. None can be replaced merely by a drawing of the limiting geometry.

Fix $p\in(1,2)$. For an integer $R\ge1$, let $N=R^4$ and form $X_R$ with $N$ agents at $A=(-1,0)$, $N$ agents at $B=(1,0)$, and one agent at $C=(0,R)$. There are $2N+1$ agents, so each coordinate median is unique.

The choice $N=R^4$ is convenient rather than essential. For the fixed range $1<p<2$, it guarantees $R^p/N\to0$ with substantial room to spare, while keeping all multiplicities integral. The proof would also work with integer multiplicities tending to infinity and satisfying that same vanishing condition. We retain the single explicit choice throughout so that every member of the family is an ordinary finite unweighted instance of exactly the mechanism being analyzed.

Odd population size removes any dependence on even-sample median conventions. Each scalar sample has a unique central order statistic, including at angles for which projections coincide. Thus the lower bound will match the upper bound for the specified mechanism and for every median-interval selection covered by the upper-bound theorem. The use of growing populations also clarifies the quantifiers: this is a worst-case ratio over all finite population sizes, rather than a claim that one fixed-size profile attains the constant exactly.

It is tempting to assume that the optimal facility is the origin because the two large clusters are symmetric. The outlier breaks that symmetry, so this assumption is not justified. Fortunately, an exact optimizer is unnecessary. Removing the outlier gives a lower bound on every facility's cost, and evaluating the full objective at the origin gives an upper bound on the optimum. The two bounds have the same leading term after normalization, which is all the ratio calculation needs.

\begin{lemma}[Optimal-cost normalization]\label{lem:opt-lower}
For this family,
\begin{equation}\label{eq:opt-sandwich}
 2N\le\OPT_p(X_R)^p\le2N+R^p,
 \qquad
 \frac{\OPT_p(X_R)}{(2N)^{1/p}}\longrightarrow1.
\end{equation}
\end{lemma}
\begin{proof}
The function $y\mapsto\norm{A-y}_2^p+\norm{B-y}_2^p$ is convex and even, so it is minimized at zero, where its value is $2$. The two large clusters therefore contribute at least $2N$ at every facility location. Placing the facility at zero gives the upper bound $2N+R^p$. Finally, $R^p/N=R^{p-4}\to0$.
\end{proof}

For completeness, convexity and evenness in Lemma~\ref{lem:opt-lower} imply the stated minimization through the midpoint inequality. If $g(y)=\norm{A-y}_2^p+\norm{B-y}_2^p$, then $g(-y)=g(y)$ and $g(0)\le(g(y)+g(-y))/2=g(y)$. Multiplying by $N$ gives the cluster contribution $2N$. After dividing the resulting sandwich by $2N$, its upper endpoint is $1+R^{p-4}/2$, which tends to one. Taking the continuous $p$th root then yields the cost normalization in the lemma.

The mechanism can nevertheless choose a point far from the origin on the scale of the cluster separation. To understand why, consider a one-dimensional sample with $N$ copies of a value $l$, $N$ copies of a value $h\ge l$, and one additional value $z$. If $z<l$, the middle observation is $l$; if $l\le z\le h$, it is $z$; if $z>h$, it is $h$. This is precisely clipping to $[l,h]$. Applying that observation in the two rotated coordinates yields the next lemma without any approximation.

\begin{lemma}[Limiting output]\label{lem:limit-output}
The URCM output on $X_R$ satisfies $\norm{m_\theta}_2\le1$ for every angle. For every fixed $0<\theta<\pi/2$,
\begin{equation}\label{eq:output-limit}
 m_\theta\longrightarrow(\cos2\theta,\sin2\theta)
 \quad\text{as }R\to\infty.
\end{equation}
\end{lemma}
\begin{proof}
For two endpoint values each repeated $N$ times and one additional value, the unique median is the additional value clipped to the interval between the endpoints. Thus the rotated coordinates of $m_\theta$ are
\begin{align*}
 a_\theta&=\clip(R\sin\theta,-|\cos\theta|,|\cos\theta|),\\
 b_\theta&=\clip(R\cos\theta,-|\sin\theta|,|\sin\theta|),
\end{align*}
where $\clip(z,l,h)=\min\{h,\max\{l,z\}\}$. It follows that
$a_\theta^2+b_\theta^2\le\cos^2\theta+\sin^2\theta=1$.
At any fixed interior first-quadrant angle, both outlier projections eventually exceed their upper endpoints, giving $a_\theta=\cos\theta$ and $b_\theta=\sin\theta$. Mapping back gives
\[
 \cos\theta\,e_\theta+\sin\theta\,f_\theta
 =(\cos2\theta,\sin2\theta).
\qedhere
\]
\end{proof}

The convergence in Lemma~\ref{lem:limit-output} has a particularly concrete meaning. At a fixed interior angle, both $\sin\theta$ and $\cos\theta$ are positive. Once $R$ exceeds both $\cot\theta$ and $\tan\theta$, the clipping constraints are saturated and the displayed limiting output is already the exact output at that angle. There is no small residual error in those coordinates. What prevents uniform saturation over the whole quadrant is that the required threshold tends to infinity near an endpoint.

Geometrically, the limiting physical output traces the upper unit semicircle as the basis angle ranges from zero to $\pi/2$. This follows from the double angle in $(\cos2\theta,\sin2\theta)$, not from any circular symmetry of the input profile. The two clusters remain fixed at opposite ends of a horizontal diameter. The trigonometric distances to those endpoints generate exactly the angular expression found in the upper bound. This agreement identifies the family as a candidate for tightness; the remaining normalization and convergence steps establish that it actually is tight.

The bound $\norm{m_\theta}_2\le1$ holds for all angles and every finite $R$, including the small angular regions where the clipping constraints have not saturated. It is therefore stronger than the pointwise limit alone. We will use it both to show that the outlier's normalized cost vanishes and to provide a single integrable bound for the entire family of angular cost functions.

\begin{proof}[Lower bound in Theorem~\ref{thm:main}]
For an odd-size sample, the unique scalar median commutes with sign changes. Increasing the basis angle by $\pi/2$ exchanges the two coordinates and changes one sign; after mapping back, the physical output is unchanged. Hence the cost on $X_R$ has period $\pi/2$, and its expectation is the average over $[0,\pi/2]$.

This periodicity is a property of the odd-sample mechanism under a change of basis. It need not be justified by a rotational symmetry of the profile, which the outlier would in fact destroy. Rotating the basis by a quarter turn only permutes the two coordinate medians and changes one coordinate sign. The unique median respects those operations, and undoing the basis change returns the same physical point. Thus averaging over one quadrant loses no part of the distribution relevant to this family.

The limiting output in Lemma~\ref{lem:limit-output} has distances $2\cos\theta$ and $2\sin\theta$ from $A$ and $B$, respectively. Define
\[
 Y_R(\theta)=\frac{C_p(X_R;m_\theta)}{(2N)^{1/p}}.
\]

The $p$th power of this normalized cost is the average of the two cluster distance powers, plus the outlier distance power divided by $2N$. In the limit, the first two terms are $(2\cos\theta)^p/2$ and $(2\sin\theta)^p/2$. Taking their sum to the power $1/p$ gives the factor $2^{1-1/p}$ in the next display. This is also why the cluster multiplicities disappear from the limiting expression: their common factor has been included in the normalization.

Since $\norm{m_\theta}_2\le1$, the outlier's contribution to $Y_R(\theta)^p$ is at most $(R+1)^p/(2R^4)$, which tends to zero uniformly in $\theta$. It follows that for almost every first-quadrant angle,
\begin{equation}\label{eq:normalized-limit}
 Y_R(\theta)\longrightarrow
 2^{1-1/p}(\cos^p\theta+\sin^p\theta)^{1/p}.
\end{equation}
Both cluster distances are at most $2$, and for $R\ge1$,
\[
 Y_R(\theta)^p\le2^p+\frac{(R+1)^p}{2R^4}
 \le2^p+2^{p-1}.
\]
This supplies an integrable bound independent of $R$.

The two requirements for dominated convergence are now explicit. Convergence holds at every interior angle and hence almost everywhere under the uniform angular distribution; the two endpoints have measure zero. Moreover, the last display bounds $Y_R$ itself by the constant $(2^p+2^{p-1})^{1/p}$ on the whole interval, for every $R\ge1$. Thus the possibly unsaturated regions near the endpoints cannot conceal a contribution that survives outside the pointwise limit. Their behavior is controlled by the same bound as the rest of the quadrant.
 Dominated convergence applied to~\eqref{eq:normalized-limit} yields
\[
 \frac{\E C_p(X_R;m_\Theta)}{(2N)^{1/p}}
 \longrightarrow\frac{2^{2-1/p}}\pi\int_0^{\pi/2}
 (\cos^p\theta+\sin^p\theta)^{1/p}\dd\theta=L(p).
\]
Together with Lemma~\ref{lem:opt-lower}, this proves
\[
 \frac{\E C_p(X_R;m_\Theta)}{\OPT_p(X_R)}\longrightarrow L(p).
\]
For every $c<L(p)$, some finite member of the family therefore has ratio greater than $c$. This establishes the matching lower bound and completes Theorem~\ref{thm:main}.
\end{proof}

The outlier is negligible in normalized social cost but determines which coordinate endpoint is selected for almost every angle in the limit. This separates its effect on the objective from its effect on the median. The construction exploits an order-statistic property of the rule, not strategic misreporting by an agent.

The final step from a limit to a worst-case lower bound is worth making explicit. Let $c<L(p)$ and choose a positive tolerance smaller than $L(p)-c$. Convergence of the finite-instance ratios implies that all sufficiently large integer values of $R$ have ratios within that tolerance of $L(p)$ and therefore greater than $c$. Every such profile has finitely many agents and finite coordinates. Consequently, no constant smaller than $L(p)$ can satisfy the approximation guarantee on all allowed profiles. Exact attainment by a finite profile is not needed for this conclusion.

The proof therefore matches the quantifiers of the upper bound without modifying the problem. It uses unweighted agents, Euclidean distance, the same social exponent, and a single uniformly chosen orientation for the mechanism. Neither fractional cluster weights nor a different distribution over angles is introduced in taking the limit. Repeated locations are legitimate reports in the model; multiplicity simply records the number of distinct agents sharing a location.

\section{The Exact Ratio and the Moment Bound}\label{sec:comparison}

The difference between $L(p)$ and $U(p)$ can be stated directly as a strict Jensen gap.

This comparison serves two purposes. First, it verifies that the upper bound proved here is numerically stronger throughout the open parameter interval, rather than merely a different representation of the previously stated bound. Second, it pinpoints the operation responsible for the difference. Both expressions depend on the same elementary angular function; they place its average and its $p$th root in different orders. Our proof retains the order appropriate to expected social cost.

The random variable in the following proof is an analytic device derived from a uniform angle. It is not an additional source of randomness in the mechanism and does not change the input profile. Its variation records the difference between directions close to the coordinate axes and directions close to the diagonals. Because that variation is nonzero for $1<p<2$, moving the concave root outside its expectation incurs a strict loss.

\begin{proposition}\label{prop:strict}
For every $1<p<2$, $L(p)<U(p)$. The integral formula extends continuously to $[1,2]$, with $L(1)=4/\pi$ and $L(2)=\sqrt2$.
\end{proposition}
\begin{proof}
Let $\Theta$ be uniform on $[0,\pi/2]$ and set $Z=\cos^p\Theta+\sin^p\Theta$. Then
\[
 L(p)=2^{1-1/p}\E Z^{1/p}.
\]
The beta integral gives
\[
 \E Z=\frac4\pi\int_0^{\pi/2}\cos^p\theta\dd\theta
 =\frac{2\Gamma((p+1)/2)}{\sqrt\pi\,\Gamma(1+p/2)},
\]
and hence $U(p)=2^{1-1/p}(\E Z)^{1/p}$. For $1<p<2$, $Z$ is not almost surely constant: its continuous defining function has value $1$ at zero and $2^{1-p/2}>1$ at $\pi/4$. Strict concavity of $z\mapsto z^{1/p}$ gives $L(p)<U(p)$.

The integrand in~\eqref{eq:L} is continuous in $(p,\theta)$ on $[1,2]\times[0,\pi/2]$ and is bounded there. The integral is therefore continuous in $p$. At $p=1$ its integral is $2$; at $p=2$ its integral is $\pi/2$. Substitution gives the endpoint values.
\end{proof}

The endpoint evaluation is a property of the formula. The endpoint approximation ratios themselves are established in the prior analyses of \citet{barak2026,chan2026}; the new theorem concerns the interior interval. The strict improvement in Proposition~\ref{prop:strict} holds throughout that interval, even though the two formulas meet at its boundaries.

The two endpoints make the source of equality transparent. At $p=1$, the root operation is linear, so moving it across an expectation creates no loss even though the angular expression varies. At $p=2$, the identity $\cos^2\theta+\sin^2\theta=1$ makes the angular expression constant, so the expectation creates no loss even though the root is concave. In the interior, neither reason for equality applies. This explains why agreement at both endpoints is compatible with a strict gap at every intermediate parameter.

It is also useful to distinguish a small numerical improvement from the logical content of an exact characterization. The size of $U(p)-L(p)$ does not determine whether a proof is needed: any positive gap leaves open which constant governs all profiles. Matching the lower bound answers that question for the specified mechanism and objective. At the same time, it does not imply a lower bound of $L(p)$ against every strategyproof mechanism. The construction is analyzed through the outputs of coordinate-wise median; a different rule could behave differently on the very same profiles.

The angular formula is directly evaluable by ordinary one-dimensional integration for any specified $p$. Numerical quadrature can illustrate the constant or check the consistency of a symbolic calculation, but it plays no role in establishing either inequality of the main theorem. The upper bound rests on inequalities valid for arbitrary finite profiles, and the lower bound rests on an explicit convergent family. This separation matters because samples of profiles or angles, however extensive, would not establish a universal approximation guarantee or its exact sharpness.

\section{Conclusion}

We determine the exact worst-case expected approximation ratio of uniformly rotated coordinate-wise median in the Euclidean plane for every $1<p<2$. The result matches the integral lower bound of Chan, Lin, and Wang. Its main ingredient is a reference-point inequality for deterministic coordinate-wise median that permits the rotation average to be taken directly at the level of social cost.

The analysis is complete for this fixed mechanism and parameter range. It does not determine the best achievable ratio among randomized strategyproof mechanisms. Two further questions concern the proof itself: whether an analogue of the reference-point inequality holds in higher dimensions, and which other objectives admit a sharp analysis by averaging a reference-point estimate. Both require additional arguments; the two-dimensional norm-ratio and half-plane estimates proved here do not by themselves establish such extensions.

The central methodological distinction is the information retained before averaging. A direction-free comparison of the two spatial norms is sufficient for a deterministic approximation bound, but it does not record how a fixed Euclidean vector is seen by a random coordinate system. The reference-point estimate keeps that directional dependence while arranging the radial coefficients to sum to the benchmark cost. The pointwise certificate and median counts justify the estimate; rotational invariance then evaluates its expectation. These are separate roles, and neither requires an explicit description of the median trajectory on an arbitrary instance.

The lower-bound family complements this perspective. Its trajectory can be calculated exactly after clipping, and its limit realizes the same angular expression. The outlier illustrates how order information can remain decisive even when its contribution to the normalized objective vanishes. Together, the two arguments explain both why the constant is an angular average and why a mean-of-powers analysis leaves a gap. Any extension of this method would need to reproduce these structural features, rather than merely replace the planar integral with a higher-dimensional average.

\section*{Disclosure of generative AI assistance}
Generative AI was used extensively in developing the mathematical arguments, carrying out computational checks, and drafting and revising this manuscript. No independent human verification or formal proof certification is claimed in this working version.

\bibliographystyle{ACM-Reference-Format}
\bibliography{references}
\end{document}